\documentclass[a4paper]{article}
\usepackage{a4wide}

\usepackage[ansinew]{inputenc}
\usepackage{graphicx}
\usepackage[spanish,english]{babel}
\usepackage{amssymb, amsthm}
\RequirePackage{lineno}

\newcommand{\R}{\ensuremath{\mathbb{R}}}
\newcommand{\NNG}{\ensuremath{\mathrm{NNG}}}

\newtheorem{theorem}{Theorem}
\newtheorem{corollary}{Corollary}
\newtheorem{obs}{Observation}
\newtheorem{lemma}{Lemma}

\usepackage {color}

\title{Minimizing interference in ad-hoc networks with bounded communication radius\footnote{A preliminary version of this paper appeared in the proceedings of the  22nd International Symposium on Algorithms and Computation (ISAAC 2011).}}

\author{Matias Korman\thanks{Universitat Polit\`ecnica de Catalunya (UPC), Barcelona. {\tt{matias.korman@upc.edu}}. With the support of the Secretary for Universities and 
Research of the Ministry of Economy and Knowledge of the Government of Catalonia and the European Union.}}

\begin{document}

\maketitle

\begin{abstract}
We consider a topology control problem in which we are given a set of sensors in $\R^d$ and we would like to assign a communication radius to each of them so that they generate a connected network and have low receiver-based interference (defined as the largest in-degree of the network). We show that any radii assignment that generates a connected network can be modified so that interference is (asymptotically) unaffected and no sensor is assigned  communication radius larger than $R_{\min}$, where $R_{\min}$  is the smallest possible radius needed to obtain strong connectivity. Combining this result with the previous network construction methods~\cite{ht-2008-miwanp,rwz-2009-amiwasn}, we obtain a way to construct a connected network of low interference and bounded radii. Since the radius of a sensor is only affected by neighboring sensors, this construction can be done in a distributed fashion.

\end{abstract}

\section{Introduction}
Ad-hoc networks are commonly used whenever a number of electronic devices are spread across a geographical area and no central communication hub exists~\cite{bghw-2008-cmn,s-2005-tcwasn}. In order to send messages between two sensors located far from each other, the message is repeated by other devices located between them. 
Due to technical constraints, the devices normally have very limited power sources (such as a small battery or solar cells). Since energy is the limiting factor for the operability of these networks, various methods have been proposed to reduce energy consumption~\cite{bp-2008-ocmiasn,ht-2008-miwanp,rwz-2009-amiwasn}. 

In most cases the transmission radius is the major source of power dissipation in wireless networks. 
Another issue that strongly affects energy consumption is interference. Intuitively, the interference of a network is defined as the largest number of sensors that can directly communicate with a single point in the plane. Indeed, lowering the interference reduces the number of package collisions and saves considerable amounts of energy which would otherwise be spent in retransmission. 

In this paper we  look for an algorithm that assigns a transmission radius to a given list of sensors in a way that the network is connected and has low interference. 
 Additionally we consider the case in which no sensor can be assigned a large radius; although theoretically one could assign an arbitrarily large radius to a sensor, in many cases this is not possible (due to hardware constraints, environmental noise, quick battery drainage, etc.). 
  We note that, in virtually all cases in which ad-hoc networks are used, sensors do not have knowledge of the location of other sensors. As a result, our aim is to give a method to construct the network in a local fashion. That is, that the communication radius of a given sensor does not depend in the location or existence of sensors that are far away. 

\section{Definitions and results}
We model each device as a point in $\R^d$ (typically $d=2$) and its transmission radius as a positive real value. Given a set $S$ of sensors, we look for a radii assignment $r:S\rightarrow \R^+$. The value $r(s)$ is the {\em communication radius} (or radius for short), and gives an idea of how far can the messages emitted from $s$ reach (or equivalently, how strong the signal is). Any radii assignment defines an undirected graph $G_r=(S,E)$ with $S$ as the ground set. The most commonly adopted model is the {\em symmetric} model in which there is an undirected edge $st\in E$ if and only if $\min\{r(s),r(t)\} \geq d(s,t)$, where $d(\cdot, \cdot)$ denotes the Euclidean distance (that is, we add an edge between two sensors if they can send messages to each other). 
Clearly, a requirement for any assignment $r$ is that the associated graph is strongly connected (i.e., for any $u,v\in S$ there is a directed path from $u$ to $v$ in $G_r$). Whenever this happens, we say that $r$ is {\em valid}. 

Asymmetric communication models have also been studied in the literature (see~\cite{fwz-iar-05,rwz-2009-amiwasn}, section IV). In this second model, any radii assignment $r$ defines a directed graph $G'_r=(S,E')$ and there is a directed edge $\overrightarrow{uv}\in E'$ whenever $r(u)\geq d(u,v)$. It is easy to observe that whenever we have $uv\in E$. In particular, $G'_r$ must be strongly connected whenever $r$ is valid. More importantly, although the focus of this paper is in the symmetric model, the same result also holds in the asymmetric case.

We model interference following the {\em received based interference} model~\cite{rwz-2009-amiwasn}. For a fixed radii assignment $r$, the {\em interference} $I(p)$ of any point $p\in\R^d$ is defined as the number of sensors that can directly communicate with $p$ (that is, $I(p)=||\{s\in S\mid r(s)\geq d(s,p)\}||$). The {\em interference} of the network (or the radii assignment) is defined as the point in $\R^d$ with largest interference.\footnote{In fact, the definition of~\cite{rwz-2009-amiwasn} only measures interference at the sensors. The extension to measuring the interference to $\R^d$ was done in \cite{ht-2008-miwanp}.} 

A more geometric interpretation of the interference is as follows: for each sensor $s\in S$ place a disk of radius $r(s)$ centered at $s$. The interference of a point is equal to the depth of that point in the arrangement of disks (analogously, the interference of the network is equal to the depth of the deepest point). This model has been widely accepted, since it has been empirically observed that most of the package collisions happen at the receiver. The interested reader can check \cite{s-2005-tcwasn} or \cite{rwz-2009-amiwasn} to see other models of  interference. 

We say that a radii assignment $r$ has bounded radius $R$ if it satisfies $r(s)\leq R$ for all $s\in S$. If we are only interested in creating a valid network with the smallest possible communication radius, the simplest approach is to consider the uniform-radius network. In this network all sensors are assigned the same radius $R_{\min}$ defined as the minimum possible value so that the associated network is strongly connected. It is easy to see that $R_{\min}$ is equal to the length of the longest edge of the minimum spanning tree of $S$. Unfortunately it is easy to see that the interference of this approach, commonly denoted by $\Delta$, can be as high as $n$ (for example, a single point located far from a large cluster of points). 

Computing the radii assignment that minimizes the interference of a given point set is NP-hard. More specifically, Buchin~\cite{b-2008-mmih} showed that it is NP-hard to obtain a valid radii assignment that minimizes the network's interference, or even approximate it with a factor of $4/3$ or less. 
  As a result, most of the previous research focuses in constructing valid networks with bounded interference, regardless of what the optimal assignment is for the given instance. 
For the symmetric 1-dimensional case (or {\em highway model}), von Rickenbach {\em et al.}~\cite{rwz-2009-amiwasn} gave an algorithm that constructs a network with $O(\sqrt{\Delta})$ interference, and showed that this algorithm approximates the minimum possible interference by a factor of $O(\sqrt[4]{\Delta})$. Afterwards, Halld\'{o}rsson and Tokuyama~\cite{ht-2008-miwanp} generalized the symmetric construction to higher dimensions, although the approximation factor does not hold anymore. Moreover, their construction uses bucketing to certify that their network has bounded radius $\sqrt{d}R_{\min}$.

A variation called the {\em all-to-one} problem was considered for the asymmetric model in~\cite{fwz-iar-05, rwz-2009-amiwasn}. In this problem, one would like to assign radii in a way that all sensors can communicate to a specific sensor $s\in S$ (called the {\em sink}). By adding a sufficiently large communication radius to $s$ we can obtain a strongly connected directed network. Note that their construction has bounded radius $R_{\min}$ for all sensors other than the sink. Unfortunately, their method needs knowledge of the whole network, and cannot be easily adapted to symmetric networks. The study of the interference generated by a random set of points was also done in~\cite{kdh-biwahnnrp-12}.

In this paper we show that any radii assignment can be transformed to another one with (asymptotically speaking) the same interference and bounded radius $R_{\min}$. In our construction, the radius of a sensor is only affected by the sensors located in its neighborhood. As a result, this network can be constructed using only local information, even if the original assignment didn't have this property.

\section{Bounded radius network}\label{sec_bounded}\label{sec_bucket}


The objective of this section is to prove the following result: 

\begin{theorem}\label{theo_boundra}
Any valid radii assignment of interference $i$ can be transformed into another assignment of interference $O(i)$ and bounded radius $R_{\min}$.
\end{theorem}

We first give an intuitive idea of our construction. The algorithm classifies the elements of $S$ into clusters. For each cluster we select a constant number of sensors (which we call the {\em leaders} of the cluster) and assign them communication radius $R_{\min}$. The main property of this set is that they are capable of sending messages to any other sensor of the cluster by only hopping through other leaders. Reciprocally, any other sensor whose radius is $R_{\min}$ will be able to communicate to a leader $\ell$ (and thus all other sensors of the cluster by hopping through $\ell)$. Hence, no sensor will ever need to have radius strictly larger than $R_{\min}$. Finally, we will connect clusters in a way that interference will not grow, except by a constant value.

More formally, the algorithm is as follows. Virtually partition the plane into $d$-dimensional cubes of side length $R_{\min}$ (each of these cubes will be referred as a {\em bucket}). For each bucket $B$, let $S_B$ be the sensors inside $B$ (i.e., $S_B= S\cap B$). Without loss of generality we can assume that no point lies in two buckets (this can be obtained by doing a symbolic perturbation of the point set). We say that two sensors $s,t\in S$ belong to the same cluster if and only if they belong to the same bucket $B$ and there is a path connecting them in the subgraph $G_u[S_B]$, where $G[S']$ denotes the subgraph of a graph $G=(S,E)$ induced by a subset of vertices $S'\subseteq S$, and $G_u$ is the network associated to the uniform radius network. 

\begin{lemma}\label{lem_clust}
For any fixed dimension $d$, there can be at most $O(1)$ clusters inside any bucket. Moreover, in each cluster there are at most $O(\Delta)$ sensors.
\end{lemma}
\begin{proof}
Partition $B$ into $d^d$ cubes of side length $R_{\min}/d$. By construction, the largest distance between any two sensors in the same sub-bucket is $\sqrt{d}R_{\min}/d=R_{\min}/\sqrt{d} \leq R_{\min}$. If there exists a bucket $B$ with more than $d^d$ clusters, we use the pigeonhole principle and obtain that there must exist at least one sub-bucket with sensors of two different clusters. In particular, there will be an edge in $G_u[S_B]$ that connects these sensors, contradicting with the definition of cluster. Proof for the second claim is identical (this time we would find a sensor whose interference in $G_u$ is larger than $\Delta$)
\end{proof}

The set of leaders of a given cluster $c$ inside a bucket $B$ is constructed as follows: as in the proof of Lemma \ref{lem_clust}, partition $B$ into cubes of sidelength $R_{\min}/d$. If all the points belong to the same sub-bucket, we pick any point as the leader. Otherwise, for any two different sub-buckets pick any edge $e$ in $G_u[c]$ connecting the two sub-buckets (if any exists) and add the vertices of $e$ to the set of leaders. We repeat this process for all pairs of sub-buckets and obtain the set of leaders.

\begin{lemma}\label{lem_leaders}
For any cluster $c$, its associated set $L_c$ of leaders has constant size. Moreover, for any sensor $s\in c$, there exists a leader $\ell_s\in L_c$ such that $d(s,\ell_s)\leq R_{\min}$.
\end{lemma}
\begin{proof}
For every pair of sub-buckets occupied by sensors of $c$, two sensors are added into $L_c$. Since a bucket is partitioned into $d^d$ sub-buckets, at most $2\times {d^d\choose 2}$ sensors will be present in $L_c$ (a constant for any fixed dimension). In order to show the second claim it suffices to observe that for any sensor $s$ there exists a leader that belongs to the same sub-bucket. 
If all sensors are located in a single sub-bucket, the claim is trivially true. Otherwise, for any given sensor $s\in c$, let $t\in c$ be any sensor located in a different sub-bucket, $\pi=(s=v_0, \ldots, v_k=t)$ be any path connecting them in $G_u[c]$, and let $i>0$ be the smallest index such that $v_i$ does not belong to the the same sub-bucket as $s$. By definition of leaders, there must be a sensor of the sub-bucket in which $v_{i-1}$ belongs to in $L_c$ (since the edge $v_{i-1}v_i$ is present in $G_u[c]$ and the two sensors belong to different sub-buckets). In particular, this leader will be within $R_{\min}$ communication distance to $v_{i-1}$ (and $s$, since they belong to the same sub-bucket). 
\end{proof}

We assign radius to all sensors $s$ of cluster $c$ as follows.  If $s\in L_c$, we assign radius $R_{\min}$. Otherwise, we assign radius equal to $\min\{r(s), R_{\min}\}$, where $r:S\rightarrow \R$ is the radii assignment of interference $i$. 

Finally, we must add a small modification to certify connectivity between clusters. We say that two clusters $c,c'$ are {\em neighboring} if there exist sensors $u\in c, v\in c'$ such that $d(u,v)\leq R_{\min}$ (we say that $u$ and $v$ are the {\em witnesses}). For any two neighboring clusters, pick any two witnesses $u,v$ and increase their radius to $R_{\min}$. Let  $\bar{r}$ be the obtained radii assignment and $G_{\bar{r}}$ its associated network. 


\begin{obs}\label{obs_neigh}
A cluster $c$ can only have a constant number of neighboring clusters. Moreover, the interference of a point $p$ in $G_{\bar{r}}$ can only be affected by sensors that are in the same or a neighboring bucket of the one containing $p$. 
\end{obs}
\begin{proof}
First notice that $c$ cannot be neighbor to another cluster of the same bucket (since it would contradict with the definition of cluster). The distance between any two points of non-neighboring buckets is larger than $R_{\min}$, hence two clusters can only be neighbors if they belong to adjacent buckets. A bucket has  $3^d-1$ neighboring buckets (a constant for fixed dimension). By Lemma \ref{lem_clust}, each such bucket can have a constant number of clusters, hence the total amount of neighboring clusters of $c$ is also bounded by a constant. The second claim is direct from the fact that no sensor is assigned radius larger than $R_{\min}$ in $G_{\bar{r}}$.
\end{proof}

\begin{lemma}\label{lem_rmin}
Any $u,u'\in S$ be two sensors such that $\bar{r}(u)=\bar{r}(u')=R_{\min}$. $u$ and $u'$ can send messages to each other, even by only using edges present in $G_{\bar{r}} \cap G_u$.
\end{lemma}
\begin{proof}
By Lemma \ref{lem_leaders}, there exist leaders $\ell, \ell'$ such that the edge $u\ell$ is in $G_{\bar{r}}\cap G_u$ (analogously $\ell'$ and the edge $u'\ell'$). Hence, it suffices to connect $\ell$ and $\ell'$. Let $\pi=(\ell=v_0, \ldots, v_k=\ell')$ be a path that connects them in $G_u$ and traverses the minimum possible number of different clusters. Let $m$ be the total number different clusters of that the path $\pi$ traverses, we will show our claim by induction on $m$. 

Consider first the case in which $m=0$; that is all sensors $v_i$ belong to the same cluster $c$ of a bucket $B$. For any $i\leq k$, let $b_i$ be the sub-bucket to which sensor $v_i$ belongs to. 
By construction of the set of leaders, each time we have $b_i\neq b_{i+1}$, there exist two leaders $t_i$ and $s_{i+1}$ such that $d(t_i,s_{i+1})\leq R_{\min}$ and belong to sub-buckets $b_i$ and $b_{i+1}$, respectively. Since both $t_i$ and $s_{i+1}$ are leaders, we must have $t_is_{i+1} \in G_{\bar{r}}\cap G_u$. Our aim is to connect $\ell$ and $\ell'$ by hopping through sensors $s_i$ and $t_i$. In order to do so we must define $s_i$ and $t_i$ for the case in which the path does not change sub-bucket. Whenever $b_i=b_{i+1}$, we simply set $t_{i}:=s_{i+1}:=s_{i}$ (we also define $s_0=\ell$ and $t_k=\ell'$). By choice of the $s_i$ and $t_i$ sensors, we always have $s_it_i \in G_{\bar{r}}\cap G_u$ for all $i\leq k$ and $t_{i}s_{i+1} \in G_{\bar{r}}\cap G_u$ for all $i< k$. 
 In particular, the path $\pi'=(\ell=s_0, t_0,s_1,t_1, \ldots, t_{k-1},s_k,t_k=\ell')$ will be feasible in $G_{\bar{r}}\cap G_u$.

Assume now that the path $\pi$ traverses different buckets. Let $i+1$ be the smallest index such that $v_{i+1}$ does not belong to the same cluster as $\ell$. By induction, the subpaths $(\ell=v_0, \ldots v_i)$ and $(v_{i+1}, \ldots , v_k=\ell')$ are feasible, hence we must connect $v_i$ with $v_{i+1}$. 

Observe that the clusters containing $v_i$ and $v_{i+1}$ are neighboring (since $v_i$ and $v_{i+1}$ are witnesses to this fact). By definition of $\bar{r}$, there will exist two sensors $w_1,w_2$ that belong to the clusters of $v_i$ and $v_{i+1}$, (respectively), and that the edge $w_1w_2$ is present in $G_{\bar{r}}\cap G_u$. We again use induction and obtain that there must exist paths $\pi_1$ (resp. $\pi_2$) connecting sensors $v_i$ and $w_1$ (resp. $v_{i+1}$ and $w_2$). Hence, by concatenating these two paths with the edge $w_1w_2$ we can connect $\ell$ with $\ell'$.
\end{proof}


\begin{lemma}\label{lem_globalconnec}
$G_{\bar{r}}$ is connected and has $O(i)$ interference.
\end{lemma}
\begin{proof}
For any two sensors $s,t\in S$, let $\pi=(s=v_0, \ldots, v_k=t)$ be the shortest path connecting them in $G_{r}$. If all edges $v_jv_{j+1}$ are present in $G_{\bar{r}}$, the path $\pi$ is feasible in $G_{\bar{r}}$ for some $j <k$. The only situation in which the edge $v_jv_{j+1}$ might not present in $G_{\bar{\mathcal{A}}}(S)$ is if the radius of $v_j$ or $v_{j+1}$ was reduced below $d(v_j,v_{j+1})$. By construction of the $\bar{r}$ assignment, this can only happen if $d(v_j,v_{j+1})>R_{\min}$. In particular, we must have $\bar{r}(v_j)=\bar{r}(v_{j+1})=R_{\min}$. In this case we can use Lemma \ref{lem_rmin} to obtain connectivity between $v_j$ and $v_{j+1}$ and proceed walking along $\pi$.

We must now show that the interference of $G_{\bar{r}}$ is indeed $O(i)$. Since no sensor is assigned radius larger than $R_{\min}$, the interference of any point $p$ can only be affected by sensors that are in the same or a neighboring bucket of the one containing $p$. There are exactly $3^d$ such buckets (a constant for fixed dimension). Combining this fact with Lemma \ref {lem_clust}, we obtain that only a constant number of buckets can affect to the interference of $p$. Hence, it suffices to see that a single cluster $c$ can only contribute a constant amount of additional interference.

Clearly, the sensors that satisfy $r(s)=\bar{r}(s)$ cannot contribute more than $i$ interference to $p$, hence we focus on the sensors of $c$ whose radii was increased. By definition of $\bar{r}$, this only happens for the set $L_c$ or sensors whose radius was increased to have connectivity with neighboring clusters. By Lemma \ref{lem_leaders} and Observation \ref{obs_neigh}, either case can only happen a constant number of times, hence the claim is shown.
\end{proof}

This completes the proof of Theorem \ref{theo_boundra}. Combining this result with the $GHUB$ network given in \cite{ht-2008-miwanp} we obtain a method to construct a network with low interference and bounded radii. 

\begin{theorem}\label{theo_llncluster}
For any set $S$ of $n$ points in $\R^d$, there exists a valid radii assignment of bounded radius $R_{\min}$ and $O(\sqrt\Delta)$ interference (for $d\leq 2$) or $O(\sqrt{\Delta\log \Delta})$ (otherwise) interference. Moreover, the bounded radius construction only adds an additional computation cost of $T(n)$, where $T(n)$ is the time needed to compute the minimum spanning tree of a set of $n$ points in $\R^d$. 
\end{theorem}
\begin{proof}
Our algorithm proceeds as follows: 
\begin{enumerate}
\item Compute $R_{\min}$, partition the sensors of $S$ into buckets and further split each bucket into clusters.
\item For each cluster $c$ apply the $GHUG$ construction algorithm, capping the maximum radius to $R_{\min}$. 
\item For each cluster also compute its set of leaders and increase their radius to $R_{\min}$.
\item For any two neighboring clusters $c$ and $c'$ pick two witnesses $s\in c$ and $s'\in c'$ and increase their radii to $R_{\min}$.
\end{enumerate}

Recall that $R_{\min}$ is equal to the length of the longest edge in the Euclidean minimum spanning tree, hence it can be computed in $O(T(n))$ time. Classifying the sensors of $S$ into buckets can esily be done in $O(n\log n)$ time for any fixed dimension. 
 For each non-empty bucket $B$ we compute the minimum spanning tree of $S_B$ and delete the edges whose length is larger than $R_{\min}$. It is easy to see that each edge removed will create an additional cluster inside the bucket. The total cost of this operation is $O(T(n_B))$, where $n_B$ is the number of points in $S_B$. Since $\sum_{S_B\neq \emptyset} n_B =n$, the total running time of the first step is bounded by $O(T(n))$.

Observe that the last two operations are very similar: in both cases we are given a collection $\mathcal{L}$ of sets of sensors, and for each pair $U,V\in \mathcal{L}$ we must find  two sensors $u\in U,v\in V$ such that $d(u,v)\leq R_{\min}$ (if any exists). Given a pair of sets $U,V$, we can check if pair of nearby sensors exists by computing the minimum spanning tree of the set $U \cup V$. One of the edges of the tree will be the closest pair between sensors of $U$ and $V$. 

This operation must be repeated for every pair of sets in $\mathcal{L}$, giving a total running time of $\sum_{i=1}^k \sum_{j\in \mathrm{Cand}(i)} O(T(m_i+m_j))=O(T(m_i)+T(m_j))$, where $m_i$ is the size of the $i$-th set, $k$ is the total number of sets in the collection, and $\mathrm{Cand}(i)$ is the set of candidate sets of the $i$-th set (i.e., sets for which there might exist sensors whose distance is at most $R_{\min}$). When computing leaders, the set $\mathrm{Cand}(i)$ is of constant size (because there are a constant number of sub-buckets). Moreover, by Observation \ref{obs_neigh}, any cluster can only be a candidate for a constant number clusters. That is, each set appears a constant number of times in the above expression. In particular, the above sum can be expressed as $\sum_{i=1}^k O(1)\times T(n_i)=O(T(n))$.  
\end{proof}

If we consider the asymmetric model, we can replace the GHUB construction and use the {\em all-to-one} construction of \cite{rwz-2009-amiwasn} instead. By doing so, we can reduce the interference.
 
\begin{corollary}
Under the asymmetric model, the interference can be reduced to $O(\log \Delta)$. 
\end{corollary}

Recall that, the {\em all-to-one} construction needs to assign large communication radius to a given sensor (so that it can send messages to all other sensors of $S$). Moreover, the exact radii assignment of a sensor would depend on non-local properties. The approach presented in here solves both issues.

\subsection*{Remarks} 
The currently best known algorithms that compute the minimum spanning tree run in $O(n\log n)$ time (for $d=2$),  $O((n\log n)^{4/3})$ (for $d=3$) or $O(n^{2-\frac{2}{(d/2+1)\epsilon}})$ (for higher dimensions)~\cite{aesw-emstbcp-91}. The time needed to construct the $GHUB$ network is polynomial, but the exact cost depends on the dimension and the computation model used (see more details in \cite{ht-2008-miwanp}). The {\em all-to-one} network can also be constructed by adding several layers of the nearest neighbor graph, each time of smaller size~\cite{rwz-2009-amiwasn}. This graph can be computed in $O(n\log n)$ time for any fixed dimension~\cite{c-faannp-83,v-aaannp-89}, hence the dominating term is $O(T(n))$ in the asymmetric model.
 
As mentioned in the Introduction, the radius assigned to a sensor $s$ in this algorithm only depends on the sensors the bucket containing $s$ and its neighboring buckets, regardless of which network model that we use. As a result, each sensor can compute its communication radius using of only local information, provided that the value of $R_{\min}$ is known in advance. Unfortunately, the value of $R_{\min}$ is a global property that cannot be easily computed in a distributed environment. Whenever this value is unknown, it can be replaced by any larger value (like for example the largest possible communication radius). By doing so we retain the local construction property and only increase the interference from $O(\sqrt{\Delta})$ or $O(\log \Delta)$ to $O(\sqrt{n})$ or $O(\log n)$, respectively.

\section{Conclusion}
Finding a method to approximate the minimum interference of any given problem instance by a $o(\sqrt{\Delta})$ factor is one of the most important open problems in this field~\cite{bp-2008-ocmiasn}. The techniques introduced in this paper provide a small step towards this goal, since they tell us that it is sufficient to construct such a network in a centralized fashion, assuming that the exact location of all sensors is known. 

The key property of our bounded radius approach is the fact that, for any cluster $c$, we can construct a set of constant size that forms a connected dominating set of $G_u[c]$. 
Although the size of this set is constant for fixed dimension, the exact value is quite large. Hence, a natural open problem is to reduce the size of this set. For $d=2$ our construction might create a set of 12 leaders for a cluster, although a more intricate proof can show that 5 sensors are always sufficient and sometimes necessary. On the negative side, we know that the exponential dependency on $d$ cannot be avoided: consider the case in which $R_{\min}=1$ and there exists a cluster with many sensors covering a unit hypercube. Any dominating set of that cluster must cover the whole cube. Since the ratio between the volume of the unit cube and the volume of the unit ball is $\frac{(d/2)!}{\pi^{d/2}} \approx \frac{\sqrt{\pi d}d^{d/2}}{(2e\pi)^{d/2}}$, at least such many sensors will be needed to dominate the cluster.


\subsection*{Acknowledgments}
{\small The author would like to thank Stephane Durocher, Maria \'{A}ngeles Garrido, Clara Grima and Alberto M\'arquez for interesting discussions on the subject, as well as the anonymous referees of ISAAC for a very thorough review.}

{\small 
\bibliographystyle{abbrv}
\bibliography{interference}}

\end{document}